\documentclass{article}
\usepackage{ifpdf}
\usepackage{cite}

\usepackage{spconf,amsmath,graphicx}

\usepackage{amsmath,amssymb,amsfonts}
\usepackage{amsthm}
\newtheorem{theorem}{Theorem}
\newtheorem{remark}{Remark}

\usepackage{algorithm}
\usepackage{algorithmic}
\usepackage{array}
\usepackage{booktabs}
\usepackage{makecell}

 \usepackage[caption=false,font=footnotesize]{subfig}

\title{Sparse Array Design for MIMO Radar in Multipath Scenarios}
\name{Xuchen Li\textsuperscript{1}, Ronghao Lin\textsuperscript{1}, Hing Cheung So\textsuperscript{2} 
\thanks{This work was supported in part by the NSFC under award
number 62271461.}}
\address{\textsuperscript{1}Department of EEIS, University of Science and Technology of China, Hefei, Anhui, China\\ \textsuperscript{2}Department of EE, City University of Hong Kong, Hong Kong, China \\ \{xuchenli@mail.ustc.edu.cn, rhlin1998@ustc.edu.cn, hcso@ee.cityu.edu.hk\}}
\begin{document}
%

\maketitle
\begin{abstract}
Sparse array designs have focused mostly on angular resolution, peak sidelobe level and directivity factor of virtual arrays for multiple-input multiple-output (MIMO) radar. 
The notion of the MIMO radar virtual array is based on the direct path assumption in that the direction-of-departure (DOD) and direction-of-arrival (DOA) of the targets are equal. However, the DOD and DOA of targets in multipath scenarios are likely to be very different. The identification of multipath targets requires DOD-DOA imaging using the the transmit and receive arrays, not the virtual array. To improve the imaging of both direct path and multipath targets, 
we introduce several new criteria for MIMO radar sparse linear array (SLA) designs for multipath scenarios. Under the new criteria, we adopt a cyclic optimization strategy under a coordinate descent framework to design the MIMO SLAs. We present several numerical examples to demonstrate the effectiveness of the proposed approaches.
\end{abstract}
\begin{keywords}
 Array design, multiple-input multiple-output radar, sparse linear array, multipath, coordinate descent
\end{keywords}
\section{Introduction}
\label{sec:intro}
Millimeter-wave radar's excellent performance in challenging weather and low-lighting conditions makes it an indispensable sensor for autonomous driving \cite{Engels2017Magazine,Patole2017Magazine,Allamd2019Magazine,Sun2020Magazine}. Multiple-input multiple-output (MIMO) radar can utilize waveform diversity to synthesize an $MN$-element virtual uniform linear array (ULA) using $M$ transmit and $N$ receive antennas \cite{Li2007Magazine,li2008mimo,garcia2022cramer}. To achieve high angular resolution, a large aperture array is required.
However, even with the MIMO technology, attaining a large aperture virtual ULA can still be difficult. To further reduce hardware cost and energy consumption, one of the most effective techniques is to use sparse linear arrays (SLAs) in lieu of ULAs \cite{Sun2020ICASSP}. 
\par The concept of sparse array design for MIMO radar typically involves deploying a reduced number of transmit and receive elements to construct a virtual SLA with a comparable aperture to that of a ULA. The array response of a sparse virtual array can vary significantly based on the placement of the transmit and receive elements, and some configurations may lead to high sidelobe or even grating lobes. Therefore, current approach focuses on selecting the positions of both transmit and receive elements to achieve a low peak sidelobe level (PSL) of the virtual SLA beampattern
\cite{Pieraccini2020radar,Liang2023TAP,Lin2022SP,genetic2017, Mateos2019EuRAD,He2022EuCAP}. 
However, these MIMO virtual SLA design methods are all for direct path targets with each target having equal direction-of-departure (DOD) and DOA. In multipath scenarios, as illustrated in Fig. \ref{fig1}, the DOD $\theta_{\mathrm{t},0}$ and DOA $\theta_{\mathrm{r},0}$ are very different. Then, the notion of virtual array falls apart \cite{Engels2017EuRAD}. To identify and remove multipath targets, the two-dimensional beamforming of the transmit and receive arrays is needed, referred to DOA versus DOD imaging \cite{Yunda2022VTC}. Therefore, it is necessary to develop new criteria to design transmit and receive SLAs for MIMO radar. To the best of our knowledge, this type of SLA design has not been considered in the literature. 
\par In this paper, we first introduce several new optimization criteria for the SLA designs for MIMO radar in multipath scenarios. Since there is no closed-form
solution to the corresponding optimization problem, we then adopt a cyclic algorithm based on the coordinate descent (CD) framework for MIMO SLA designs. Numerical examples are provided to illustrate the effectiveness of the SLA designs. 
\par $Notation$: We use bold lowercase letters to denote vectors. $(\cdot)^T$ and $(\cdot)^H$ represent the transpose and conjugate transpose operators, respectively. $\otimes$ denotes the Kronecker product. $\{ \cdot \}$ represents a set of elements from a vector, and $| \{ \cdot \}|$ denotes the cardinality of the set.

\section{Problem Formulation}
\label{sec:format}
Consider a  MIMO radar with a linear transmit array consisting of $M$ antennas and a linear receive array consisting of $N$ antennas. The transmit and receive arrays are positioned in a horizontal parallel configuration. We allow the inter-element spacing to be integer multiples of $\frac{\lambda}{2}$, where $\lambda$ represents the wavelength. The locations of the transmit and receive antennas normalized by the $\frac{\lambda}{2}$ are denoted, respectively, as: 
\begin{equation}
\begin{split}
     &\mathbf{x}_{\mathrm{t}} = \left[ x_{\mathrm{t},1}, x_{\mathrm{t},2},\cdots, x_{\mathrm{t},M} \right]^T, \\
     &\mathbf{x}_{\mathrm{r}} = \left[ x_{\mathrm{r},1}, x_{\mathrm{r},2},\cdots, x_{\mathrm{r},N} \right]^T.
\end{split}   
\end{equation}
Note that the entries of $\mathbf{x}_{\mathrm{t}}$ and $\mathbf{x}_{\mathrm{r}}$ are all integers.  To simplify the subsequent analysis, we let $x_{\mathrm{t},1}=0$ and $x_{\mathrm{r},1}=0$. The steering vectors of the transmit and receive arrays can be expressed, respectively, as:
\begin{equation}\label{steer}
\begin{split}
    &\mathbf{a}_{\mathrm{t}}(\theta_{\mathrm{t}}) = \left[ 1,e^{j\pi x_{\mathrm{t},2}\sin\theta_{\mathrm{t}}},\cdots,e^{j\pi x_{\mathrm{t},M}\sin\theta_{\mathrm{t}}} \right]^{T}, \\
    & \mathbf{a}_{\mathrm{r}}(\theta_{\mathrm{r}}) = \left[1,e^{j\pi x_{\mathrm{r},2}\sin\theta_{\mathrm{r}}},\cdots,e^{j\pi x_{\mathrm{r},N}\sin\theta_{\mathrm{r}}} \right]^{T},\\
\end{split}
\end{equation}
where $\theta_{\mathrm{t}}$ and $\theta_{\mathrm{r}}$ are DOD and DOA, respectively. Consider a particular target with DOD and DOA, denoted by $\theta_{\mathrm{t},0}$ and $\theta_{\mathrm{r},0}$, respectively, impinging on the receive array. 
The standard two-dimensional the DOA versus DOD image of the target has the form
\begin{equation}\label{spatial}
\footnotesize
    P\left( \theta_{\mathrm{t}},\theta_{\mathrm{r}},\theta_{\mathrm{t},0},\theta_{\mathrm{r},0} \right)= \frac{1}{MN} \big| \big(\mathbf{a}_{\mathrm{t}}(\theta_{\mathrm{t}}) \otimes \mathbf{a}_{\mathrm{r}}(\theta_{\mathrm{r}})\big)^{H}
    \big(\mathbf{a}_{\mathrm{t}}(\theta_{\mathrm{t,0}}) \otimes \mathbf{a}_{\mathrm{r}}(\theta_{\mathrm{r,0}})\big) \big|.
\end{equation}
Because the inter-element spacing of the arrays are integer multiples of $\frac{\lambda}{2}$, the imaging possesses the shift invariant property \cite{Lin2022SP}, which means that as $\theta_{\mathrm{t},0}$ and $\theta_{\mathrm{r},0}$ vary from $-90^{\circ}$ to $90^{\circ}$, the DOA versus DOD image experiences only cyclic shifts, while the characteristics of the sidelobes remain unchanged. Therefore, we focus on optimizing the arrays for the case of $\theta_{\mathrm{t},0} = 0^{\circ}$ and $\theta_{\mathrm{r},0} = 0^{\circ}$. Note that $\mathbf{a}_{\mathrm{t}}(\theta_{\mathrm{t}}) \otimes \mathbf{a}_{\mathrm{r}}(\theta_{\mathrm{r}}) $ degenerates into a MIMO virtual array steering vector when $\theta_{\mathrm{t}} = \theta_{\mathrm{r}}$. Let $\mathbf{x}_{\mathrm{v}} = [x_{\mathrm{v},1},x_{\mathrm{v},2},\cdots,x_{\mathrm{v},L}]^{T}$ represent the locations of the elements in the virtual array. We have
\begin{figure}[t]
\centering
\includegraphics[scale=0.8]{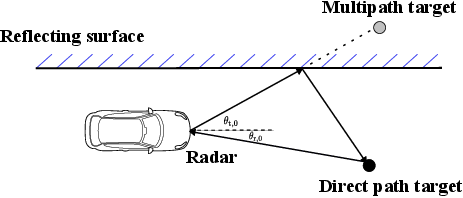}
\caption{A multipath scenario.}
\label{fig1}
\end{figure}
\begin{equation}
      \{ \mathbf{x}_{\mathrm{v}} \} = \{ \mathbf{x}_{\mathrm{t}} \} + \{ \mathbf{x}_{\mathrm{r}} \}, \\
\end{equation}
where the elements in $ \{ \mathbf{x}_{\mathrm{v}} \}$ are formed by adding the corresponding elements from \{$\mathbf{x}_{\mathrm{t}}\}$ and $\{ \mathbf{x}_{\mathrm{r}} \}$, and $|\{ \mathbf{x}_{\mathrm{v}} \}| = L$ denotes the number of unique elements in the MIMO virtual array. Let $\mathbf{w}_{\mathrm{v}} = [w_{1},w_{2},\cdots,w_{L}]^{T}$
denote the corresponding weight vector. For example, $w_{2}=2$ implies that there are two virtual elements at $x_{\mathrm{v},2}$. The normalized virtual array beampattern can be expressed as
\begin{equation}
    P_{\mathrm{v}} \left( \theta,0^{\circ} \right) = \frac{1}{MN} |\mathbf{w}^{T}\mathbf{a}_{\mathrm{v}}(\theta) |,
\end{equation}
where $\mathbf{a}_{\mathrm{v}}(\theta) = [e^{j \pi x_{\mathrm{v},1}},\cdots,e^{ j \pi x_{\mathrm{v},L}}]^{T} $ is the steering vector of the resulting virtual array. Note that the virtual array beampattern is important for the high resolution spacial imaging of direct path targets. When only optimizing the virtual array beampattern for SLA designs, it is desirable to ensure non-repetition of the virtual array elements, i.e., $L = MN$, and $\mathbf{w} = \mathbf{1}^{T}$. A common strategy involves setting the minimum inter-element spacing of the transmit array to match the entire aperture of the receive array \cite{Hamza2021ICASSP, Cohen2018ICASSP, Sun2020ICASSP}. However, this results in grating lobes in the DOD dimension of the two-dimensional DOA versus DOD image for the imaging of multipath targets. To simultaneously ensure the performance of the DOA versus DOD imaging for multipath targets and the virtual array spatial imaging for direct path targets, we consider the following array design mertics: resolutions of DOD and DOA, PSL of DOA versus DOD image, directivity factor (DF) of the virtual array beamformer. Accordingly, we consider the optimization criteria as follows:
\begin{itemize}
    \item We constrain the DOD resolution, $\delta_{\mathrm{t}}$, and the DOA resolution, $\delta_{\mathrm{r}}$, within some desired values of $\delta_1$ and $\delta_2$, respectively. 
    \item We minimize the the peak sidelobe level (PSL) of the DOA versus DOD image.
    \item We limit the number of distinct elements in the MIMO virtual array such that $ L > \sigma MN$,  where $\sigma \leq 1$ is an adjustable parameter.
\end{itemize}
The first criterion ensures the resolutions of DOA and DOD, while the second one indicates that the PSL of DOA versus DOD imaging serves as the objective function. The last criterion implies that we allow for partial repetition of the virtual array elements. According to following theorem, the presence of duplicated elements in the virtual array leads to DF degradations, and is related to $\sigma$.
\begin{theorem}\label{Theorem1}
     The DF of the virtual array is maximized at $10 \log_{10}(MN)$ if and only if $L=MN$.  When $L > \sigma MN$, the DF loss does not exceed $10\log_{10}(\sigma+(1-\sigma)^2(MN))$.
\end{theorem}
\begin{proof}
    See Appendix.
\end{proof}
\begin{remark}\label{remark1}
    Note that the upper bound on the maximum loss of DF in Theorem 1 assumes that all repeated array elements are positioned at the same location. In practical cases, the actual loss is often much smaller than this upper bound.
\end{remark}
In summary, we formulate the optimization problem of the SLA design for MIMO radar in multipath scenarios as follows:
\begin{align}\label{opt}
&\min_{\mathbf{x}_{\mathrm{t}},\mathbf{x}_{\mathrm{r}}} \max_{\theta_{\mathrm{t}},\theta_{\mathrm{r}}\in \Omega} P(\theta_{\mathrm{t}},\theta_{\mathrm{r}},0^{\circ},0^{\circ}) \nonumber 
\\
&\quad \quad {\rm s.t.} \quad \delta_{\mathrm{t}} < \delta_{1}, \delta_{\mathrm{r}} < \delta_{2}, \nonumber
\\
& \  \qquad \qquad L > \sigma MN, \nonumber
\\
& \  \qquad \qquad  \mathbf{x}_{\mathrm{t}} \in \mathbb{N}^{M \times 1} ,\mathbf{x}_{\mathrm{r}} \in \mathbb{N}^{N \times 1},
\end{align}
where $\Omega$ represents the two dimensional sidelobe region of the spatial spectrum. The optimization problem in (\ref{opt}) is non-convex, and it is NP-hard in general and difficult to solve.

\section{Sparse Linear Array Design}
\label{sec:pagestyle}

\begin{algorithm}[t]
    \caption{Cyclic Algorithm for MIMO SLA Design}
    \begin{algorithmic}[1]
        \REQUIRE{$\delta_1,\delta_2,\sigma$ and $\Omega$, $\mathbf{x}_{\mathrm{t}}, \mathbf{x}_{\mathrm{r}}$ as a feasible pair of transmit and receive arrays. Initialize $\text{PSL}=0$.}
        \ENSURE{ $\mathbf{x}_{\mathrm{t}},\mathbf{x}_{\mathrm{r}}$ with low PSLs in DOA and DOD images.}
        \STATE  Calculate  $\text{PSL}_{\mathrm{t}}, \text{PSL}_{\mathrm{r}}$, and determine $\text{PSL}_{0}$. 
    \WHILE{$\text{PSL}_{0}< \text{PSL}$ }
    \STATE $\text{PSL}=\text{PSL}_{0}$
    \FOR{ $m = 2, \cdots, M-1$}
        \FOR{$x_{\mathrm{t},m} \in \mathcal{X}_{\mathrm{t},m}$}
        \IF{$L > \sigma MN$}
        \STATE Calculate $\text{PSL}_{\mathrm{t}}$ and store it.
        \ENDIF
        \ENDFOR
      \STATE  Set $\text{PSL}_{\mathrm{t}}^{\text{new}}$ to be the minimum $\text{PSL}_{\mathrm{t}}$ and update $x_{\mathrm{t},m}$ to the corresponding position.
     \ENDFOR
    \STATE Repeat 4 to 11 for $\mathbf{x}_{\mathrm{r}}$ and obtain $\text{PSL}_{\mathrm{r}}^{\text{new}}$.
     \STATE $\text{PSL}_{0} = \max\{ \text{PSL}_{\mathrm{t}}^{\text{new}},\text{PSL}_{\mathrm{r}}^{\text{new}}\}$
    \ENDWHILE
    \end{algorithmic}
\end{algorithm}

In this section, we convert (\ref{opt}) into an antenna selection problem that can be solved using the CD framework. Substituting (\ref{spatial}) into (\ref{opt}) and utilizing the properties of the Kronecker product yields
\begin{equation} \label{obj}
\max_{\theta_{\mathrm{t}},\theta_{\mathrm{r}} \in \Omega} \frac{1}{MN} | \mathbf{1}^{T} \mathbf{a}_{\mathrm{t}}(\theta_{\mathrm{t}}) \cdot \mathbf{1}^{T} \mathbf{a}_{\mathrm{r}}(\theta_{\mathrm{r}}) | .
\end{equation}
Using the Cauchy-Schwarz inequality, we get 
\begin{equation}\label{ieq2}  \mathbf{1}^{T}\mathbf{a}_{\mathrm{t}}(\theta_{\mathrm{t}}) \leq M, \quad \mathbf{1}^{T}\mathbf{a}_{\mathrm{r}}(\theta_{\mathrm{r}}) \leq N.
\end{equation}
Employing (\ref{ieq2}), we can equivalently express (\ref{obj}) as
\begin{equation}
        \max \{ \max_{\theta_{\mathrm{t}}\in \Omega_{\mathrm{t}}} \frac{1}{M} | \mathbf{1}^{T}\mathbf{a}_{\mathrm{t}}(\theta_{t}) | , \max_{\theta_{\mathrm{r}}\in \Omega_{\mathrm{r}}} \frac{1}{N} | \mathbf{1}^{T}\mathbf{a}_{\mathrm{r}}(\theta_{\mathrm{r}})|    \}  ,
\end{equation}
where $\Omega_{\mathrm{t}}$ and $\Omega_{\mathrm{r}}$ denote the sidelobe regions of the DOD and DOA beampatterns, respectively. Next, we consider the constraint variables of the optimization problem (\ref{opt}). The resolutions $\delta_{\mathrm{t}}$ and $\delta_{\mathrm{r}}$ for DOD and DOA, respectively, can be approximately calculated using the half-wavelength normalized aperture $D_{\mathrm{t}}$ and $D_{\mathrm{r}}$ of the transmit and receive arrays as \cite{Sun2020Magazine}:
\begin{equation}
\begin{split}
     \delta_{\mathrm{t}} = 2\arcsin \left ( \frac{2.8}{\pi D_{\mathrm{t}}} \right), \\
    \delta_{\mathrm{r}} = 2\arcsin \left( \frac{2.8}{\pi D_{\mathrm{r}}} \right),
\end{split}
\end{equation}
To satisfy the resolution constraints, we fix the last elements of the transmit and receive arrays, i.e., $x_{\mathrm{t},M} = D_{\mathrm{1}} $ and $x_{\mathrm{r},N} = D_{\mathrm{2}}$, where $D_{1} >  5.6 / (\pi \sin\delta_{\mathrm{t}}) $ and $D_{2} >  5.6/ (\pi \sin\delta_{\mathrm{r}}) $. Since the performance of the DOD and DOA dimensions can be separately optimized, we use the CD framework to convert the original problem (\ref{opt}) into two subproblems. 
In the first subproblem, we fix the receive array, and select one element from the transmit array, excluding the first and last elements, as a variable. The locations of the other transmit array elements remain unchanged. This leads to Subproblem 1:
\begin{equation}
\begin{split}
& \min_{x_{\mathrm{t},m}} \quad \max_{\theta_{\mathrm{t}} \in \Omega_{\mathrm{t}}} \frac{1}{M}|\mathbf{1}^{T}\mathbf{a}_{\mathrm{t}}(\theta_{\mathrm{t}})|
\\
&\quad \quad \quad {\rm s.t.} \quad x_{\mathrm{t},m} \in \mathcal{X}_{\mathrm{t},m}, L > \sigma MN,
\\
&\quad \quad \quad \quad \quad \mathcal{X}_{\mathrm{t},m} = \{0,\cdots,D_{1}\} \backslash \{ \mathbf{x}_{\mathrm{t},\sim m} \}, 
\end{split}
\end{equation}
where $x_{\mathrm{t},m}$ denotes the current variable of interest, and $\mathcal{X}_{\mathrm{t},m}$ represents the set of available locations for $x_{\mathrm{t},m}$. Similarly, we apply the same procedure to the receive array, resulting in Subproblem 2 as:
\begin{equation}
\begin{split}
& \min_{x_{\mathrm{r},n}} \quad  \max_{\theta_{\mathrm{r}} \in \Omega_{\mathrm{r}}} \frac{1}{N}|\mathbf{1}^{T}\mathbf{a}_{\mathrm{r}}(\theta_{\mathrm{r}})|
\\
&\quad \quad \quad {\rm s.t.} \quad x_{\mathrm{r},n} \in \mathcal{X}_{\mathrm{r},n}, L > \sigma MN,
\\
&\quad \quad \quad \quad \quad \mathcal{X}_{\mathrm{r},n} = \{0,\cdots,D_{2}\} \backslash \{ \mathbf{x}_{\mathrm{r},\sim n} \},  
\\
\end{split}
\end{equation}
where $x_{\mathrm{r},n}$ denotes the current variable of interest, and $\mathcal{X}_{\mathrm{r},n}$ represents the set of available locations for $x_{\mathrm{r},n}$. It can be seen that the optimization of the original problem (\ref{opt}) can be achieved by alternatingly optimizing Subproblems 1 and 2 in a cyclic algorithm. Given the desirable parameters $\delta_{1},\delta_{2}$ and $\sigma$, we begin with initializing a feasible pair of transmit and receive arrays. We calculate the PSL values related to the transmit and receive arrays, denoting them as $\text{PSL}_{\mathrm{t}}$ and $\text{PSL}_{\mathrm{r}}$, respectively. The larger of $\text{PSL}_{\mathrm{t}}$ and $\text{PSL}_{\mathrm{r}}$ is referred as $\text{PSL}_{0}$. For the transmit array, we cyclically update the elements from $x_{\mathrm{t},2}$ to $x_{\mathrm{t},M-1}$. For each element, we sequentially place it in the available locations and calculate a new $\text{PSL}_{\mathrm{t}}$. If the new $\text{PSL}_{\mathrm{t}}$ is lower than the previous value, we adjust the position of the element to the new one until $\text{PSL}_{\mathrm{t}}$ is minimized. Subsequently, an identical procedure is performed on the receive array until $\text{PSL}_{\mathrm{r}}$ is minimized. The larger of the minimized $\text{PSL}_{\mathrm{t}}$ and $\text{PSL}_{\mathrm{r}}$ is the updated $\text{PSL}_{0}$. If the new value of $\text{PSL}_{0}$ reduces, we repeat the entire aforementioned process until $\text{PSL}_{0}$ no longer decreases. The proposed steps are summarized in Algorithm 1. Note that Algorithm 1 can be executed in parallel with different initializations to improve the computational efficiency.


\section{Numerical Examples}
\label{sec:typestyle}

We now provide several numerical examples to evaluate the performance of our SLA design for MIMO radar in multipath scenarios. Note that all the PSL values are verified using the analytical solutions from \cite{Lin2022SP}.
\par Consider $M=10$ transmit and $N=10$ receive antennas, with the requirement that the DOD and DOA resolutions are both higher than $\delta_{1} = \delta_{2} =  2^{\circ}$. Consequently, we set both the transmit and receive array apertures as $D_{1}=D_{2}=60$ half-wavelengths. The grid points for $\theta_{\mathrm{t}}$ and $\theta_{\mathrm{r}}$ are both set to $K=1000$. The parameter $\sigma$, which determines the number of distinct elements in the virtual array, is set to $0.7$.
\par We randomly initialize $2000$ sets of independent and feasible transmit and receive array locations, and then use Algorithm 1. Then we select the transmit and receive arrays corresponding to the minimum PSL as the final optimized solution. The optimized positions of the MIMO SLA including transmit and receive SLAs, as well as the locations of the  virtual SLA elements including the corresponding weight vector, are shown in Fig. \ref{fig2}. 
\begin{figure}[h]
\centering
\includegraphics[scale=0.5]{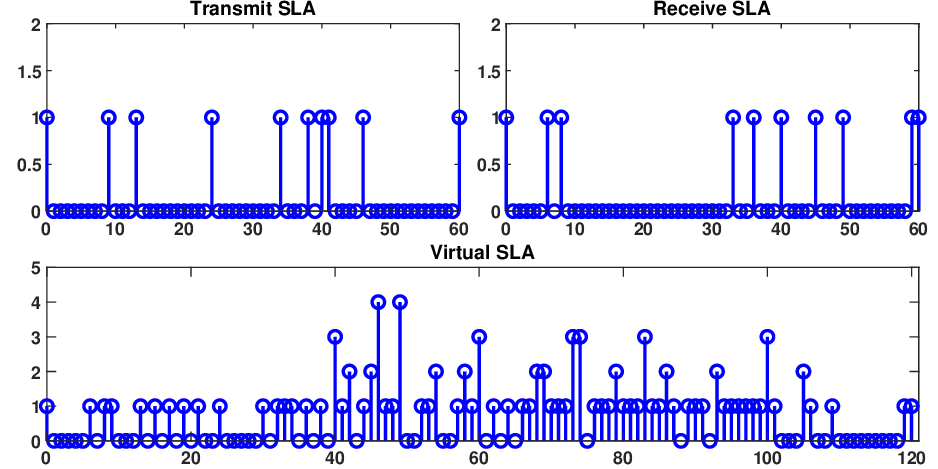}
\caption{Optimized MIMO SLA configurations}
\label{fig2}
\end{figure}
\par The DOA versus DOD image of a multipath target with DOD=$-20^{\circ}$ and DOA=$30^{\circ}$ is plotted in Fig. \ref{fig3}\subref{3a}. To verify the effectiveness of the algorithm, we compare the PSL values of the initial random MIMO SLAs and the optimal MIMO SLA. It can be observed that the optimized PSL is superior to all the initial values in Fig. \ref{fig3}\subref{3b}.
\begin{figure}[h]
    \centering
    \subfloat[]{
      \includegraphics[width = 0.465\linewidth]{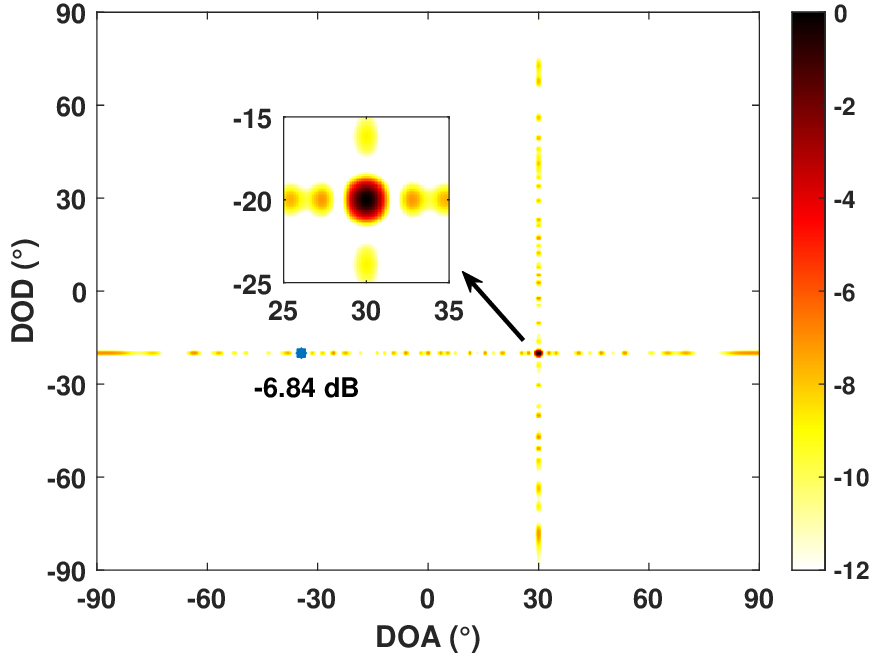} 
    \label{3a} }\hfil
    \subfloat[]{
     \includegraphics[width = 0.465\linewidth]{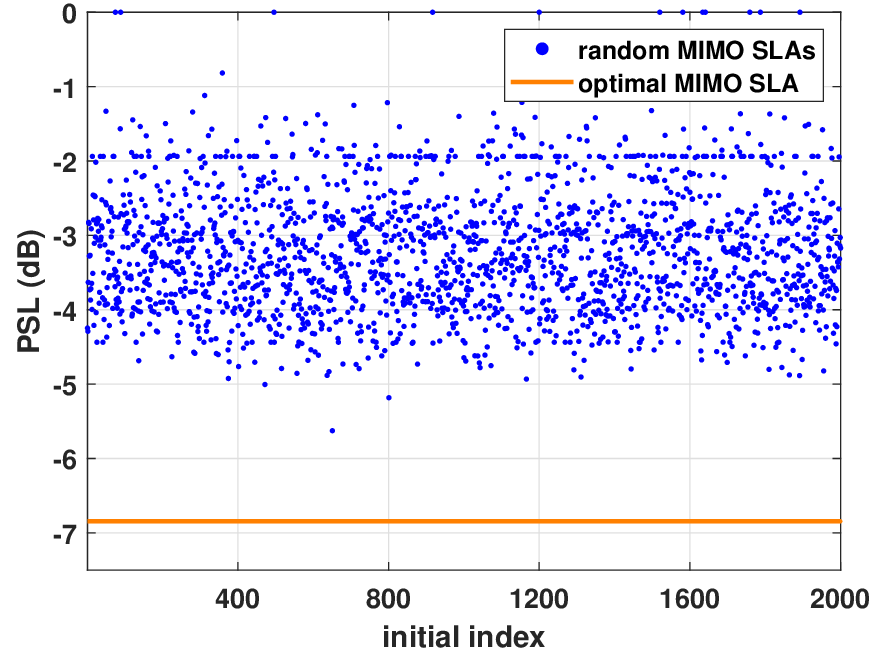}
    \label{3b} } \hfil
    \caption{\protect\subref{3a} DOA versus DOD image. \protect\subref{3b}  Comparison of PSL between optimized MIMO SLA and random MIMO SLAs.}
    \label{fig3}
\end{figure}

Next, we compare the performance of the optimized MIMO SLA, referred to as SLA$^{\star}$, with well-known sparse arrays including the minimum redundancy array (MRA) and nested array (NA) \cite{Moffet1968TAP,Shi2021TAES}. The numbers of elements of both transmit and receive arrays are 10 for all configurations. For the NA, both the inner ULA and the outer ULA have 5 array elements. We focus on the metrics including the PSL of the DOA versus DOD image, $\delta_{\mathrm{t}}$, $\delta_{\mathrm{r}}$, and the DF of the virtual array, as shown in Table.\ref{tabel1}. It can be seen that our optimized
MIMO SLA outperforms the others in terms of all metrics. Specifically, with $\sigma = 0.7$, the upper bound of DF loss is approximately $9.87$ dB, while the actual DF loss of the optimized virtual array is $2.55$ dB. This result is consistent with the explanation in Remark \ref{remark1}.

\begin{table}[t] 
    \centering
    \caption{DOA versus DOD imaging performance of different arrays}
    \renewcommand{\arraystretch}{1.3}
    \scalebox{1.0}{
    \begin{tabular}{c|cccc}
          \hline
    \bf{Arrays} & \bf{PSL} (dB)& $\boldsymbol{\delta}_{\mathrm{t}} (^{\circ})$  & $\boldsymbol{\delta}_{\mathrm{r}} (^{\circ}) $ & \bf{DF} (dB) \\
    \hline
    MRA &  -5.86 &  2.10 &  2.10 & 16.38\\
    NA &   -5.32 & 3.08 & 3.08  &  15.09\\
    \bf{SLA$^{\star}$} &  \bf{-6.84}  & \bf{1.66}  & \bf{1.40}  & \bf{17.45} \\
     \hline
    \end{tabular}
        }
\label{tabel1}
\end{table}
\section{Conclusion}
\label{conclusion}
We have established several new criteria for MIMO radar SLA designs with high angular resolution and low PSL in multipath scenarios. Based on these criteria, we have formulated a novel optimization problem and derived the maximum potential DF loss of the optimized virtual SLA.  
We have presented a cyclic algorithm based on the CD framework to solve the optimization problem.
The effectiveness of the proposed algorithm has been verified using numerical examples.
\begin{appendix}
    The DF of the (5) is defined as
\begin{equation}\label{DF}
    \mathrm{DF} = \frac{P(0^{\circ},0^{\circ})^2}{\int_{-90^{\circ}}^{90^{\circ}}P(\theta,0^{\circ})^2 \cos\theta d\theta}.
\end{equation}
According to \cite{Lin2022SP}, (\ref{DF}) can be expressed as
\begin{equation}\label{DF1}
    \mathrm{DF} = \frac{| \sum_{l=1}^{L}w_{l} |^2}{\sum_{l=1}^{L}w_{l}^{2}} = \frac{(MN)^2}{\sum_{l=1}^{L}w_{l}^{2}}.
\end{equation}
Then the maximum DF value is achieved only when $w_{l} = 1$ for $l=1,\cdots,L$, as $\mathrm{DF}_{\mathrm{max}} = MN$. Assume that there are only $\sigma MN$ elements equal to 1 in the weighting vector $\mathbf{w}$. For the denominator of (\ref{DF1}), because $w_{l}\geq 0$, we have 
\begin{equation}\label{ieq}
    \sum_{l=1}^{L} w_{l}^{2} \leq \sigma MN + (1-\sigma)^2(MN)^2.
\end{equation}
Combining (\ref{DF1}) and (\ref{ieq}) yields:
\begin{equation}
    \frac{\mathrm{DF}(\sigma)}{\mathrm{DF}_{\mathrm{max}}} \geq \frac{1}{\sigma+(1-\sigma)^2(MN)}.
\end{equation}
The proof is complete.
\end{appendix}
\label{appendix}

\vfill\pagebreak


\bibliographystyle{IEEEbib}



\end{document}